\documentclass[11pt]{article}

\usepackage{enumitem}
\usepackage{amsfonts}
\usepackage{amsmath}
\usepackage{amsthm}
\usepackage[dvipsnames]{xcolor}
\usepackage[a4paper, margin=1in]{geometry}
\usepackage{booktabs} 
\usepackage[ruled]{algorithm2e} 

\usepackage{natbib}
\usepackage{authblk}

\usepackage{hyperref}

\newcommand{\bigparen}[1]{\left( #1 \right)}
\newcommand{\bigbraces}[1]{\left\{ #1 \right\}}

\newcommand{\fdiff}[2]{\frac{\mathrm{d} #1}{\mathrm{d} #2}}
\newcommand{\realnum}{\mathbb{R}}
\newcommand{\argmin}{\operatornamewithlimits{arg\,min}}
\newcommand{\argmax}{\operatornamewithlimits{arg\,max}}

\setlist[itemize]{noitemsep, topsep=0pt}

\newtheorem{theorem}{Theorem}
\newtheorem{lemma}{Lemma}
\newtheorem{corollary}{Corollary}[lemma]

\theoremstyle{definition}
\newtheorem{definition}{Definition}
\theoremstyle{definition}
\newtheorem{example}{Example}

\title{Proportional Dynamics in Linear Fisher Markets with Auto-bidding: Convergence, Incentives and Fairness}

\date{}

\author{Juncheng Li}
\author{Pingzhong Tang}
\affil{Tsinghua University}

\allowdisplaybreaks

\begin{document}
	
\maketitle
	
	\begin{abstract}
		Proportional dynamics, originated from peer-to-peer file sharing systems, models a decentralized price-learning process in Fisher markets.
		Previously, items in the dynamics operate independently of one another, and each is assumed to belong to a different seller.
		In this paper,
		we show how it can be generalized to the setting where each seller brings \textit{multiple} items and buyers allocate budgets at the granularity of \textit{sellers} rather than individual items.
		The generalized dynamics consistently converges to the competitive equilibrium, and interestingly relates to the auto-bidding paradigm currently popular in online advertising auction markets.
		In contrast to peer-to-peer networks, the proportional rule is not imposed as a protocol in auto-bidding markets.
		Regarding this incentive concern, we show that buyers have a strong tendency to follow the rule, but it is easy for sellers to profitably deviate (given buyers' commitment to the rule).
		Based on this observation, we further study the seller-side deviation game and show that it admits a unique pure Nash equilibrium.
		Though it is generally different from the competitive equilibrium, we show that it attains a good fairness guarantee as long as the market is competitive enough and not severely monopolized.
	\end{abstract}

\section{Introduction}

Market dynamics describes the interplay among market members in response to various signals.
For a solution concept of a model, besides its economic properties, if there exists some natural and computationally efficient dynamics that leads to it, it is much more likely to emerge in real-world markets.
The modern theory of general equilibrium is pioneered by \citet{walras1900elements} with the well-known dynamics named \textit{t\^atonnement process}.
It captures the intuitive phenomenon where the price increases if the demand is greater than supply, and lowers otherwise.
However, t\^atonnement does not specify a rule to determine intermediate off-equilibrium outcomes \citep{branzei2021proportional}.
Buyers interact with items by reporting their demands given the current pricing, and it is when at least an approximate equilibrium is reached that trade really happens.
The requirement of a centralized coordinator who reports prices and collects demands is also not always realistic \citep{cole2008fast,dvijotham2022convergence}.
Moreover, the discrete version of the process may not even converge to the competitive equilibrium in the foundational linear Fisher market \citep{goktas2023t}, which models the classic scenario where a set of items is to be sold to a set of budget-constrained buyers with linear utilities.

Besides t\^atonnement process, another dynamics that received a lot of attention in the last decade is the \textit{proportional dynamics} (or proportional response dynamics).
In this process, at each round, buyers are required to bid for every item and pay their bids in advance.
After all the bids are collected, each item is allocated to buyers \textit{in proportion to their bids}.
Conversely, buyers' bids are generated by allocating budgets among items \textit{in proportion to the utility} they receive from each item in the last round.
In contrast to t\^atonnement process, proportional dynamics operates in a fully decentralized way, requires no parameter configuration, and always converges to the competitive equilibrium in linear Fisher markets with a fast rate \citep{zhang2011proportional}.

Despite its mathematical attractiveness, the two-sided proportional allocation rule is less common to be observed in the real world (especially compared to the natural price update rule of t\^atonnement,  \cite{zhang2011proportional}), apart from the peer-to-peer (P2P) file sharing systems from which it is motivated \citep{wu2007proportional}.
In this paper, we generalize proportional dynamics to the case where each seller comes with \textit{multiple} items and buyers allocate budgets on the basis of sellers rather than items.
In particular, our new formulation bears a close resemblance to the \textit{online advertising} markets with first price auction and auto-bidding, which puts it in a context of great practical significance.

In today's online advertising auction markets, instead of directly competing with each other in every auction, advertisers could only participate through proxy \textit{auto-bidders} provided by the seller.
Within an advertising platform, the seller has a strong control over all the components, while buyers are largely price-taking.
Importantly, with budget-constrained value-maximizing auto-bidders, the online advertising auction market is exactly a linear Fisher market, and the steady-state named \textit{pacing equilibrium} formed by auto-bidders also has a deep relationship to the competitive equilibrium of the linear Fisher market \citep{conitzer2022pacing,conitzer2022multiplicative}.

Motivated by this connection, we incorporate the auto-bidding auction model into proportional dynamics, and explore three research questions.
First, to justify that our generalization is reasonable and meaningful, we need to establish that the dynamics converges consistently to the competitive equilibrium while maintaining a reasonable convergence rate.
Second, from the perspective of algorithmic game theory, we study the incentives for buyers and sellers to adhere to or deviate from the proportional allocation rule.
Finally, since the incentive analysis shows that buyers are more inclined to follow the rule while sellers are prone to deviate, we are interested in finding out which state the strategic behaviors of sellers might bring the market to, and whether a certain degree of fairness could still be guaranteed on the buyer side.

\subsection{Our Results}

In this paper, we extend proportional dynamics to the one-seller multi-item setting, which to the authors' knowledge has not been studied in the literature.
The generalized dynamics works as follows.
At the beginning of each round, each buyer reallocates its budget among sellers in proportion to the utility it receives from each seller in the last time step.\footnote{As in the original formulation, to make it work properly, every buyer should initially put a non-zero amount of money on \textit{every} seller.}
Within each sub-market that consists of all the items owned by a single seller, items are sold through simultaneous first price auctions (one for each item).
There is an auto-bidder provided by the seller using the \textit{multiplicative pacing} strategy to bid on behalf of each buyer, i.e., the bid for every item must have the form $\alpha v$, where $\alpha$ is the \textit{pacing multiplier} chosen by the auto-bidder to pace the rate at which the budget is spent and will be uniformly applied to \textit{all} items, and $v$ is the valuation of the item to the buyer.
Given the amount of money paid in advance to the seller,
the allocation is determined by presuming that auto-bidders reach the pacing equilibrium \citep{conitzer2022pacing}, at which the pacing multiplier of every buyer is simultaneously optimal in hindsight.
The pacing equilibrium is unique with respect to the utility received by each buyer, therefore the dynamics is fully deterministic.

We first show that, if both sides of the market truthfully implement the rule specified above, the dynamics always converges to the competitive equilibrium.
The one-seller multi-item setting brings new difficulties, and we will see in the proof how  pacing equilibrium helps us overcome them and hence justifies itself as a meaningful generalization of the per-item proportional rule.
For the rate of convergence, we show that the average \textit{Eisenberg-Gale objective}\footnote{This objective function is maximized at the competitive equilibrium and relates to the well-known fairness measure, the Nash social welfare. Therefore it serves as a reasonable measure of the distance from an off-equilibrium state to the competitive equilibrium. Its definition can be found in Section \ref{subsec:fisher_markets_and_competitive_equilibrium}.} function converges to the optimal value at a rate of $O(T^{-1})$, where $T$ denotes the number of rounds.

We proceed to explore the incentive of buyers and sellers to implement the (generalized) proportional rule.
For buyers, though not always optimal, the proportional rule gives a 2-approximation to the optimal utility. Moreover, the exact utility maximization problem is both non-convex and non-smooth, and its optimization requires non-local information that is hard to acquire in reality.\footnote{Besides, in practice, even knowing its sub-optimality, equalizing bang-per-bucks (or revenue-on-investments) among sellers is widely accepted and easily interpretable, particularly for  non-expert buyers.} Therefore buyers have a strong motivation to follow the proportional rule.

On the other hand, sellers have readily available instruments to deviate from the vanilla version of first price auction.
If the mechanism of \textit{additive boosts} is allowed (which exists almost everywhere in real-world auto-bidding markets and is widely studied in literate; see, e.g., \cite{balseiro2021robust}), sellers can subsidy or penalize auto-bidders at each auction while a generalized pacing equilibrium can still be reached.
Moreover, given any budget profile submitted by buyers, it is easy for the seller to compute proper boosts to make \textit{any} allocation as the generalized pacing equilibrium within its sub-market.
Therefore sellers can directly manipulate the allocation to compete for buyers’ budgets, anticipating their proportional allocating behavior.
	We show that the competitive equilibrium gives an incentive ratio of 5 to each seller, i.e., a unilateral deviation from the competitive equilibrium could bring the manipulator a revenue of at most 5 times the original.
	Nonetheless, the revenue optimization problem is convex and the information requirement is easy to meet. Therefore sellers are still prone to deviate.
	
Based on the incentive analysis, we formulate a seller-side game to study the consequence of seller deviation when buyers commit to the proportional rule.
	We show that the game always admits a unique pure Nash equilibrium (PNE).
	Though it generally does not agree with the competitive equilibrium, it guarantees a Nash social welfare at least a $(1 - \Delta)$ fraction of the optimal, where $\Delta \in (0, 1]$ is a parameter characterizing the degree of \textit{monopolization} of the market.
	Therefore as long as the market is not severely monopolized, the fairness of the competitive equilibrium can be largely preserved at the PNE.

\subsection{Related Work}

Our work sits at the intersection of two important lines of literature, namely the study on proportional dynamics in Fisher markets and on auto-bidding in online advertising auction markets.
Researchers in the former area typically stick to the setup where each seller brings only one item.
To the authors' knowledge, our work is the first to put the dynamics in a context where each seller brings multiple items.
In contrast, in auto-bidding markets each seller (advertising platform) naturally brings multiple items.
Importantly, \citet{conitzer2022pacing,conitzer2022multiplicative} relates the auto-bidding markets with either first or second price auction to linear Fisher markets.
They find that, with first price auction, the pacing equilibrium formed by auto-bidders can be viewed as a slightly adapted version of the competitive equilibrium,\footnote{The difference lies in that they study budget-constrained \textit{quasi-linear} utility-maximizing advertisers, while in Fisher markets value-maximizing buyers are considered (as done in this paper). To the authors' knowledge, value-maximizers significantly outweigh quasi-linear utility-maximizers in trading volume nowadays.} while with second price auction, it should additionally assume that each buyer is supply-aware.
These works constitute the theoretical foundation for us to incorporate auto-bidding into proportional dynamics in Fisher markets.
Nonetheless, \citet{conitzer2022multiplicative, conitzer2022pacing} consider the market owned by a single seller without outside competitors, while we extend it to the setting where multiple sellers compete with each other.
\citet{paes2020competitive} and \citet{despotakis2021first} also consider the competition among multiple auto-bidding platforms, in an attempt to explain the recent trend in the industry that more and more platforms shift from second price to first price auction, but their models are not built upon Fisher markets.
Below we will review related works from different lines of literature in more detail.

\paragraph{Online advertising auction markets with auto-bidding.}

Today auto-bidders are not yet diverse and powerful enough to optimize any utility function.
Linear-utility maximization with budget-constraint (the type of auto-bidder studied in this paper; see, e.g., \cite{gao2021online}) is one of the most adopted options.
Moreover, with limited feedback from the seller, typically buyers could only make decisions based on its payment and acquired value on each seller, even though they are not truly linear-utility maximizers.
As a result, linear Fisher markets serve as a good first approximation to online advertising markets and a starting point for more complex models \citep{conitzer2022multiplicative,conitzer2022pacing,li2023vulnerabilities}.
Another constraint also frequently used in practice is Return-On-Investment (ROI, or equivalently Return-On-Ad-Spend, ROAS) \citep{balseiro2021robust,golrezaei2021auction}, and sometimes both budget and ROI constraints are imposed \citep{aggarwal2019autobidding,deng2021towards}.
The multiplicative pacing strategy is one of the most implemented strategies in the industry.
It is applicable for both budget and ROI constrained bidders and for both first and second price auction \citep{conitzer2022multiplicative,conitzer2022pacing,li2023vulnerabilities}.

Both first and second price auction are usually augmented with the mechanism of additive boosts, which exists almost everywhere in the industry and is also widely studied in the literature.
Previously researchers \citep{deng2021towards,balseiro2021robust} mainly focus on the effect of boosts within the sub-market owned by a single seller.
Our work complements this line of research by showing how it could be utilized in the competition among multiple platforms.

First price auction, the auction format studied in this paper, is gaining popularity in recent years in the industry.
Researchers also show that, when auto-bidders apply the multiplicative pacing strategy, first price auction outperforms second price in many regards \cite{conitzer2022pacing, li2023vulnerabilities}, which endows sellers with a strong motivation to adopt it.
In particular, the pacing equilibrium can be computed efficiently in large-scale first price auction markets \cite{gao2021online}, which renders our dynamics computationally practical.
Though second price auction cannot be fully captured in our framework, with additive boosts and reserve prices \cite{balseiro2021robust}, sellers can manipulate the allocation in a similar way as dictated in Section \ref{sec:incentive_sellers}.
Therefore our results could also provide some insights into it.

\paragraph{Trading post game and proportional dynamics.}
The item-side proportional allocation rule is first studied by
\citet{shapley1977trade}. They propose the trading post game to study markets involving non-price-taking buyers, where each item comes with a trading post to handle  bid collection and item allocation.
The rule also emerges in the mechanism design literature on various resource allocation settings.
\citet{kelly1997charging} shows that, in a computer network, if the capacity of each link is allocated in proportion to the bid of each demand, the social welfare can be maximized at equilibrium.
\citet{feldman2005price} apply the rule in allocating computing resources and show that the Nash equilibrium attains a good performance on both efficiency and fairness.

\citet{wu2007proportional} first propose the proportional dynamics for a special case of Fisher markets modeling P2P file sharing systems such as BitTorrent. 
Within this particular context, each individual is simultaneously a buyer and a seller, and the proportional rule is indeed implemented as the protocol.
Note that, though a P2P network seems to be an exchange (Arrow-Debreu) market, this first proportional dynamics does not directly generalize to the exchange model \citep{branzei2021proportional} since there is no money involved.
After \citet{zhang2011proportional} generalizes the dynamics to general Fisher markets, researchers further extend it to various settings and establish its convergence \cite{cheung2018dynamics,cheung2018tracing,branzei2021proportional}.

\paragraph{Computation of competitive equilibrium.}
The computation of solution concepts is an ongoing topic in the study of economic models.
For Fisher markets, the tractability of computing competitive equilibrium largely depends on the utility functions of buyers.
For linear utilities, besides decentralized dynamics,  the competitive equilibrium can also be computed in a centralized way via convex programs \citep{eisenberg1959consensus, shmyrev2009algorithm} or combinatorial algorithms \citep[e.g.,][]{orlin2010improved}.
For other families of utilities, e.g., additively separable and PLC utilities, the problem becomes PPAD-hard \citep{chen2009spending}.
These complexity results put a limit on what we could expect from a decentralized dynamics, particularly on the rate of convergence.
In this paper, we show an \textit{ergodic} rate\footnote{I.e., the rate at which the average of some value (instead of the last iterate) converges to the limit point.}
$O(T^{-1})$ of convergence for the Eisenberg-Gale objective, which is identical to the rate established by \citet{branzei2021proportional} for proportional dynamics in Arrow-Debreu markets.
Previously the fast convergence of proportional dynamics is established mainly via interpreting the dynamics as a decentralized mirror descent process over a market-wise convex program \citep{birnbaum2011distributed,cheung2018dynamics}.
It is interesting to study whether such a connection exists for our one-seller multi-item setting or the Arrow-Debreu markets.

\paragraph{Nash social welfare and fairness.}
Nash social welfare (NSW) is a well-established measure of fairness in budget-based allocation.
In linear Fisher markets, the competitive equilibrium is characterized by the Eisenberg-Gale program \citep{eisenberg1959consensus} that optimizes the NSW.
When valuations are unknown to the market owner, \citet{cole2013mechanism} propose a truthful mechanism for buyers to achieve a 2.718-approximation of the optimal NSW.
The Shapley-Shubik trading post game, though untruthful, is shown to give a 2-approximation of NSW at any PNE, and it also guarantees proportionality for every individual buyer \citep{branzei2022nash}.
Note that the game considered in Section  \ref{sec:seller_competition_game} captures the competition among \textit{sellers}, while the aforementioned works focus on the buyer-side strategic behaviors.
Our work is quite distinct from previous results, and it is interesting to see that buyers can provide themselves a certain degree of fairness by collectively committing to the proportional allocation rule.

\section{Model Description}
\label{sec:model_description}

\subsection{Linear Fisher Markets and Competitive Equilibrium}
\label{subsec:fisher_markets_and_competitive_equilibrium}

A \textbf{linear Fisher market} consists of a set $J$ of $m$ divisible items (each with a normalized 1 unit of supply) and a set $I$ of $n$ buyers.
Each buyer $i$ has a budget $B_i$.
The utility function $u_i$ of buyer $i$ has the form
$
u_i(x) =
\sum_j v_{i, j} x_{i, j},
$
where $x_{i, j} \in [0, 1]$ is the fraction of item $j$ allocated to buyer $i$, and $v_{i, j} \geq 0$ is the value of item $j$ to buyer $i$.
For each item, there is at least one buyer $i$ with $v_{i, j} > 0$.
The total budgets of all buyers is denoted as $B = \sum_i B_i$.

A \textbf{competitive equilibrium} $(p^*, x^*)$ consists of a pricing $p^* \in [0, +\infty)^m$ and an allocation $x^* \in [0, 1]^{n \times m}$ that satisfies:
\begin{itemize}
	\item Each buyer gets allocated an optimal bundle it affords:
	
	$
	x_i^* \in \argmax_{x_i \in [0, +\infty)^{m}: \sum_j p_j x_{i, j} \leq B_i} \sum_j v_{i, j} x_{i, j};
	$
	\item The market clears:
	$
	\sum_i x_{i, j} = 1, \forall j.
	$
\end{itemize}
The competitive equilibrium is also known to satisfy the following properties.
\begin{itemize}
	\item $x$ is an equilibrium allocation if and only if it solves the \textbf{Eisenberg-Gale (EG) program}:
	\begin{align*}
		\max \quad & \sum_i B_i \ln u_i
		\\
		\text{s.t.} \quad &
		u_i \leq \sum_j v_{i, j} x_{i, j}, \forall i;
		\\
		&
		\sum_i x_{i, j} \leq 1, \forall j;
		\\
		& x_{i, j} \geq 0, \forall i, j.
	\end{align*}
	$\sum_i B_i \ln u_i$ is called the \textbf{EG-objective}.
	Note its relationship with our fairness measure NSW, defined as follows:
	\begin{displaymath}
		\text{NSW}(x) = \prod_{i} \bigparen{u_i(x)}^{\frac{B_i}{B}}.
	\end{displaymath}
	\item An equilibrium always exists, and is unique w.r.t. prices and buyers' utilities.
	Buyers always deplete their budgets at equilibrium.
	\item Fixing $p$, the \textbf{bang-per-buck} of item $j$ for buyer $i$ is $\frac{v_{i, j}}{p_j}$, the value received by buyer $i$ for each unit of money paid for item $j$.
	At equilibrium, every buyer gets its most wanted items with respect to the bang-per-buck, i.e., if $
	x_{i, j} > 0$,
	then
	$\frac{v_{i, j}}{p_j}
	\geq
	\frac{v_{i, j'}}{p_{j'}}, \forall j'$.
\end{itemize}

\subsection{Auto-bidding Markets and Pacing Equilibrium}

The traditional Fisher market model is \textit{seller-oblivious}, where sellers do not have the power to affect the price and the item-ownership has no effect on the competitive equilibrium.
In this paper, we assume that the items are partitioned into several sub-markets, each owned by a single seller $k$.
The set of all sellers is denoted as $K$.
We will consider the incentive issues of sellers and how their strategic behaviors would influence the market outcome.

For each sub-market owned by a single seller, we assume that items are sold through simultaneous first price auctions (one auction for one item) with an auto-bidding mechanism.
Each item is allocated to the buyer with the highest bid, and the winner should pay its bid.\footnote{Items are allowed to be allocated fractionally, as is commonly done in the auto-bidding literature.}
In contrast to bidding for each item directly,
in an auto-bidding market,
each buyer only submits in advance the amount of money it is willing to spend on the seller.
The bidding is instead handled by an auto-bidder (one for each buyer) using the \textit{multiplicative pacing} strategy, i.e., there is a \textit{pacing multiplier} $\alpha_i > 0$ for each buyer $i$ such that buyer $i$'s bid for item $j$ is $\alpha_i v_{i, j}$.
The steady-state reached by the auto-bidders is called a (first price) \textbf{pacing equilibrium} $(\alpha, x)$ \citep{conitzer2022pacing}, which consists of pacing multipliers $\alpha$ and allocation $x$ such that: (1) item $j$ is sold at price $p_j = \max_i \alpha_i v_{i, j}, \forall j$; (2) if $x_{i, j} > 0$, $\alpha_i v_{i, j} = p_j$; (3) $\sum_i x_{i, j} = 1, \forall j$; (4) $\sum_j p_j x_{i, j} = B_i, \forall i$.
The rule of first price auction is enforced by conditions (1-3), while condition (4) shows that each auto-bidder best responds to its opponents.
To see the latter, note that both the received utility and the payment of each auto-bidder are monotone with respect to the pacing multiplier $\alpha_i$, therefore depleting the budget is always the optimal solution.

Note that, given buyers' submitted budgets, the sub-market owned by each seller is itself a linear Fisher market.
The pacing equilibrium actually captures the same state (regarding allocation and payment) to the competitive equilibrium of this sub-market: here $p_j$ is exactly the price at the competitive equilibrium, and $\alpha_i$ is the inverse of the maximum bang-per-buck over all items.
Though technically equivalent, pacing equilibrium brings a new perspective to the classic Fisher market model, and even inspires new algorithms computing the competitive equilibrium at a large scale \citep{gao2021online}.
We refer readers to \citep{conitzer2022pacing} and its related works for more information.
In particular, though multiplicative pacing  is not optimal in hindsight for each buyer at the level of individual auctions, the resulting overall allocation is in many regards better than the second price pacing equilibrium \citep{conitzer2022pacing} that guarantees each buyer the ex-post optimality.
Combined with the on-going trend in the industry that platforms keep moving from second to first price auction, it is practically significant to understand the price-forming dynamics among competing sellers using first price auction and auto-bidding.

\subsection{Proportional Dynamics in Auto-bidding Markets}

With the introduction of auto-bidding, the generalized proportional dynamics is formally defined as follows.
\begin{definition}
	The \textbf{proportional dynamics} in auto-bidding markets is an iterative process where buyers repeatedly allocate their budgets among sellers and sellers allocate items via pacing equilibrium.
	Let $B_i(k, t)$  be the amount of money buyer $i$ allocates to seller $k$ at time $t$, and $u_i(k, t)$ be the utility received by buyer $i$ at the pacing equilibrium
	of the sub-market owned by seller $k$.
	Let $u_i(K, t) = \sum_k u_i(k, t)$.
	Buyers adopt the following update rule:
	$
	B_i(k, t + 1) = \frac{u_i(k, t)}{u_i(K, t)} \cdot B_i.
	$
\end{definition}

With our definition, the original proportional dynamics \citep{zhang2011proportional} is the special case where each seller owns only one item.
The buyer-side proportional rule follows the same spirit of the original formulation (with only a change of granularity), and in Section \ref{sec:proportional_dynamics} we will see why allocation via pacing equilibrium is a meaningful generalization of the proportional rule for sellers.
Note that for each buyer $i$, at time $t$, the only information required to update its budget allocation is $\{u_{i}(k,t)\}_{k \in K}$, which will always be known to it in the real world.

Below we will denote by $B_i(k, *)$ the money of buyer $i$ spent on seller $k$ at an arbitrary (but fixed) competitive equilibrium, and $u_i(K, *)$ the (unique) utility of buyer $i$ at the competitive equilibrium.
We will also use $x_{i, j}(t)$ and $x_{i, j}(*)$ to denote the allocation at time $t$ and at the competitive equilibrium, respectively.
When the time is not relevant or could be easily derived from the context, we will use $u_i(k)$ and $u_i(K)$ to denote the utility received by buyer $i$ from seller $k$ and all sellers, respectively.

\section{Convergence  of Proportional Dynamics}
\label{sec:proportional_dynamics}

In this section, we justify that the proportional dynamics still converges to the competitive equilibrium with a reasonable rate.
The proof centers around a potential function $\Phi(t)$ that generalizes the one used in establishing convergence in the one-seller one-item setting (for linear Fisher markets \citep{zhang2011proportional} and linear exchange markets \citep{branzei2021proportional}):
\begin{displaymath}
	\Phi(t)
	=
	\sum_{i, k} B_i(k, *) \ln \frac{B_i(k, *)}{B_i(k, t)},
\end{displaymath}
i.e., the KL-divergence between the budget allocation at the competitive equilibrium and the one at time $t$.
The main task is to show that the potential keeps decreasing over time.
We will see in the proof how the pacing equilibrium helps us overcome the difficulties brought by the one-seller multi-item setting.


\begin{theorem} \label{thm:convergence_non_personalized_pd}
	A competitive equilibrium is a fixed point of proportional dynamics, and
	proportional dynamics always converges to a competitive equilibrium.
\end{theorem}

\begin{proof}
	The first half of the theorem is straightforward to verify.
	
	To establish convergence, we relate $\Phi(t)$ with $\Phi(t + 1)$ as follows:
	\begin{align*}
		B_i(k, *) \ln \frac{B_i(k, *)}{B_i(k, t + 1)}
		=&
		B_i(k, *) \ln \bigparen{ \frac{B_i(k, *)}{B_i} \cdot \frac{u_i(K, t)}{u_i(k, t)} }
		=
		B_i(k, *) \ln \bigparen{ \frac{u_i(k, *)}{u_i(K, *)} \cdot \frac{u_i(K, t)}{u_i(k, t)} }
		\\
		=&
		B_i(k, *) \ln \bigparen{\frac{B_i(k, *)}{B_i(k, *)} \cdot \frac{B_i(k, t)}{B_i(k, t)} \cdot \frac{u_i(k, *)}{u_i(K, *)} \cdot \frac{u_i(K, t)}{u_i(k, t)} }
		\\
		=&
		B_i(k, *) \ln \frac{B_i(k, *)}{B_i(k, t)}
		- B_i(k, *) \ln \frac{u_i(K, *)}{u_i(K, t)}
		- B_i(k, *) \ln \bigparen{\frac{B_i(k, *)}{B_i(k, t) \cdot \frac{u_i(k, *)}{u_i(k, t)}}},
	\end{align*}
	where the second equality holds since at the competitive equilibrium, the bang-per-bucks of all received items for a seller are equalized.\footnote{Regarding the risk of division by zero, note that $B_i(k, t)$ and $u_i(k, t)$ can become zero only if all items owned by seller $k$ have a zero value for buyer $i$. In this case, $B_i(k, *)$ must also be zero, and the corresponding terms can be safely ignored.}
	Summing over buyers and sellers gives
	\begin{align*}
		\Phi(t) - \Phi(t + 1)
		= &
		\sum_{i, k} B_i(k, *) \ln \frac{u_i(K, *)}{u_i(K, t)} +
		\sum_{i, k}  B_i(k, *) \ln \bigparen{\frac{B_i(k, *)}{B_i(k, t) \cdot \frac{u_i(k, *)}{u_i(k, t)}}}.
	\end{align*}
	For the first term, it holds that
	\begin{displaymath}
		\sum_{i, k} B_i(k, *) \ln \frac{u_i(K, *)}{u_i(K, t)}
		=
		\sum_i B_i \ln \frac{u_i(K, *)}{u_i(K, t)} \geq 0,
	\end{displaymath}
	since $u = u(K, *)$ maximizes the EG-objective $\sum_i B_i \ln u_i$.
	For the second term,\footnote{In the one-seller one-item setting \citep{zhang2011proportional}, based on the per-item proportional rule, the second term can be directly simplified to the KL-divergence between $p(*)$ and $p(t)$, which denotes the item prices at the equilibrium and at time $t$, respectively. This cannot be done here. We instead use the property of pacing equilibrium to bound the term, and thus justify that it is a non-trivial generalization of the per-item proportional rule.} let $\alpha_i(k, t) = \frac{B_i(k, t)}{u_i(k, t)}$ be the pacing multiplier for buyer $i$ at the pacing equilibrium of sub-market $k$.
	We have
	\begin{align*}
		\tilde{B} := & \sum_{i, k} B_i(k, t) \cdot \frac{u_i(k, *)}{u_i(k, t)}
		= 
		\sum_{k} \sum_i \alpha_i(k, t) \sum_{j \in J_k} v_{i, j} x_{i, j}(*) \\
		\leq &
		\sum_{k} \sum_i \alpha_i(k, t) \sum_{j \in J_k} v_{i, j} x_{i, j}(t)
		= 
		\sum_{k} \sum_i B_i(k, t)
		= B,
	\end{align*}
	where the inequality holds since, at the pacing equilibrium, every item is sold to the buyer with the maximum bid $\alpha_i v_{i, j}$.
	Then
	\begin{align*}
		 \sum_{i, k}  B_i(k, *) \ln \bigparen{\frac{B_i(k, *)}{B_i(k, t) \cdot \frac{u_i(k, *)}{u_i(k, t)}}} 
		\geq 
		\sum_{i, k}  B_i(k, *) \ln \bigparen{\frac{  B_i(k, *)}{ \frac{B}{\tilde{B}} \cdot B_i(k, t) \cdot \frac{u_i(k, *)}{u_i(k, t)}}}
		\geq 0,
	\end{align*}
	where the second inequality holds due to the non-negativity of KL-divergence.
	
	Now we have shown that $\Phi(t)$ is monotonically decreasing.
	By the non-negativity of KL divergence, $\Phi(t)$ is lower bounded.
	As a result, $\Phi(t)$ converges and $\Phi(t) - \Phi(t + 1)$ goes to zero as $t$ goes to infinity.
	Recall that $\Phi(t) - \Phi(t + 1) \geq \sum_{i, k} B_i(k, *) \ln \frac{u_i(K, *)}{u_i(K, t)} \geq 0$ and the equality only happens when $u_i(K, t) = u_i(K, *), \forall i$ by the uniqueness of utilities at the competitive equilibrium.
	Therefore $u_i(K, t) \rightarrow u_i(K, *)$ as $t \rightarrow \infty$.
	KL-divergence equals zero only if two distributions are identical, so we have
	$B_i(k, t) \cdot \frac{u_i(k, *)}{u_i(k, t)} \rightarrow B_i(k, *)$, i.e., the bang-per-bucks $\frac{u_i(k, t)}{B_i(k, t)}$  converge to $\frac{u_i(k, *)}{B_i(k, *)}$, the bang-per-bucks at the competitive equilibrium, and so do the item prices.
	
	To see the convergence of $B_i(k, t)$, first take the limit point $B_i(k, **)$ of some convergent sub-sequence (which exists as the sequence is bounded), then replace the $B_i(k, *)$ in the definition of $\Phi$ to $B_i(k, **)$.
	Now we have $\Phi(t)$ converges to 0 and thus $B_i(k, t)$ converges to $B_i(k, **)$ (if $\Phi(t)$ is defined with some $\{B_i(k, *)\}_{i \in I, k \in K}$ differing from the limit point, it will rest at some value larger than zero and this cannot rule out the possibility that $B_i(k, t)$ oscillates among several different equilibrium values).
\end{proof}

For the convergence rate, the techniques used in \cite{branzei2021proportional} are still applicable, which shows that the average EG-objective converges at a rate of $O(1/T)$.
To see this, by summing both sides of the inequality $\Phi(t) - \Phi(t + 1) \geq \sum_{i, k} B_i(k, *) \ln \frac{u_i(K, *)}{u_i(K, t)}$ over $t$, we have
$
\Phi(1) - \Phi(T + 1)
\geq
\sum_{t = 1}^T \bigparen{ \sum_{i} B_i \ln u_i(K, *) - \sum_i B_i \ln u_i(K, t) }.
$
Dividing both sides by $T$ gives the result.

\section{Incentive of Buyers}

Considering that sellers implement the proportional rule truthfully, we examine the incentive of buyers to deviate and focus on the buyer-side game defined as follows.

\begin{definition}
	A \textbf{budget allocation game} of a linear Fisher market is a normal form game among buyers.
	Each buyer $i$ allocates an amount $B_i(k)$ of money to seller $k$.
	The budget allocation should be feasible: $\sum_k B_i(k) \leq B_i, \forall i$.
	Items in $J_k$ are sold according to the allocation at  the pacing equilibrium with budgets $\{B_i(k)\}_{i \in I}$.
	The utility of each buyer is the total valuation it receives.
\end{definition}

In Appendix \ref{app:example_buyer_optimization}, we give an example to show that buyer's optimization problem could be both non-convex and non-smooth.
In addition, to learn the utility curve, information such as budgets and valuations of other buyers is required.
Without such knowledge, the buyer can only access the utility function point-by-point through repeated interaction with sellers, during which other buyers' budget allocation might change as well.
As a result, it is impractical for the buyer to manipulate its budget allocation and optimize its utility exactly.\footnote{In comparison, when each seller owns only one item, a profitable deviation is easily computable for buyers. See Appendix \ref{app:incentive_buyers_original} for a discussion.}

On the other hand, the proportional rule is not only easily implementable, but also provides a satisfying performance guarantee as shown below.

\begin{theorem} \label{thm:proportional_rule_buyer_guarantee}
	Fixing other buyers' budget allocation over sellers, if buyer $i$'s bang-per-bucks $\{\frac{u_i(k)}{B_i(k)}\}_{k \in K}$ are equalized among all sellers, then the corresponding budget allocation $\{B_i(k)\}_{k \in K}$ achieves a utility of at least $1/2$ of the optimal.
\end{theorem}

To prove the theorem, we will use the monotonicity of the bang-per-buck and the received utility at the pacing equilibrium with respect to the buyer's submitted budget, which are known  results in the literature (see Appendix \ref{app:budget_monotonicity} for references).

\begin{lemma} \label{lemma:monotonicity_budget}
	Assume that seller $k$ truthfully implements the pacing equilibrium.
	Fixing other buyers' submitted budgets to seller $k$, buyer $i$'s received utility $u_i(k)$ is weakly increasing with respect to $B_i(k)$, while the bang-per-buck ratio $\frac{u_i(k)}{B_i(k)}$ is weakly decreasing with respect to $B_i(k)$.
\end{lemma}

\begin{proof}[Proof of Theorem \ref{thm:proportional_rule_buyer_guarantee}]
	Assume that, except for buyer $i$, all the other buyers' budget allocations are fixed.
	Let $\{B_i(k)\}_{k \in K}$ be the budget allocation of buyer $i$ such that its bang-per-bucks are equalized across all sellers, and $u_i(k)$ be the utility received from seller $k$ by submitting $B_i(k)$ to it.
	Let $u_i = \sum_k u_i(k)$.
	
	Consider any other feasible budget allocation $\{\tilde{B}_i(k)\}_{k \in K}$.
	Let $\tilde{u}_i(k)$ be the utility received from seller $k$ by submitting $\tilde{B}_i(k)$ to seller $k$, and $\tilde{u}_i = \sum_k \tilde{u}_i(k)$.
	Partition $K$ into $K_{\leq}$ and $K_{>}$ such that  for each $k \in K_{\leq}$, $\tilde{B}_i(k) \leq B_i(k)$, and for each $k \in K_>$, $\tilde{B}_i(k) > B_i(k)$.
	
	By the first result of Lemma \ref{lemma:monotonicity_budget}, for each $k \in K_{\leq}$, buyer $i$ spends no more money, so it receives no more utility: $\tilde{u}_i(k) \leq u_i(k)$.
	By the second result of Lemma \ref{lemma:monotonicity_budget}, for each $k \in K_>$, buyer $i$ spends more money, so its bang-per-buck cannot be higher: $\frac{\tilde{u}_i(k)}{\tilde{B}_i(k)} \leq \frac{u_i(k)}{B_i(k)} = \frac{u_i}{B_i}$.
	Combining both parts together, we have:
	\begin{align*}
		\tilde{u}_i
		=
		\sum_{k \in K_{\leq}} \tilde{u}_i(k)
		+ \sum_{k \in K_{>}} \tilde{u}_i(k)
		\leq
		\sum_{k \in K_{\leq}} u_i(k)
		+ \sum_{k \in K_>} \tilde{B}_i(k) \cdot \frac{u_i}{B_i}
		\leq
		u_i + u_i = 2u_i.
	\end{align*}
\end{proof}


\section{Incentive of Sellers}
\label{sec:incentive_sellers}

%

In real-world markets, the additive boost mechanism is widely implemented by sellers to subsidize or penalize each buyer in each auction.
With additive boosts $c$, buyers are ranked by the boosted bids $\alpha_i v_{i, j} + c_{i, j}$, but if it wins, the payment is still the non-boosted bid $\alpha_i v_{i, j}$.
Note that adding a constant to all $c_{i, j}$ will not change the auction outcome.
Below we assume $c_{i, j} \geq 0, \forall i, j$.
A \textbf{pacing equilibrium with additive boosts} $c \in \realnum^{n \times m}$ consists of pacing multipliers $\alpha$ and allocation $x$ such that: (1) $p_j = \max_i \alpha_i v_{i, j} + c_{i, j}, \forall j$; (2) if $x_{i, j} > 0$, $\alpha_i v_{i, j} + c_{i, j} = p_j$; (3) $\sum_j (p_j - c_{i, j}) x_{i, j} = B_i, \forall i$; (4) $\sum_i x_{i, j} = 1, \forall j$.
Similar to the original version, pacing equilibrium with additive boosts can be characterized by a modified EG program. The result is given as Lemma \ref{thm:pacing_equilibrium_with_boosts} in Appendix \ref{app:pacing_equilibrium_with_boosts}.
We then show that, if each seller allocates items via the pacing equilibrium with fixed boosts throughout the dynamics, the convergence result remains to hold.

\begin{theorem} \label{thm:convergence_with_boosts}
	With fixed additive boosts $c$, proportional dynamics always converges to a (market-wise) pacing equilibrium with additive boosts $c$.
\end{theorem}
The proof is given in Appendix \ref{app:convergence_with_boosts}.
Observe that, given the submitted budgets $\{B_i(k)\}_{i \in I}$, seller $k$ can implement \textit{any} feasible allocation as the pacing equilibrium by setting proper additive boosts: with the target allocation in mind, the pacing multiplier is given by the inverse of the bang-per-buck, and then boosts can be calculated such that the winner has the highest boosted bid.
Therefore a seller can directly choose a feasible allocation within its sub-market to attract buyers' proportionally allocated budgets.
The following theorem characterizes how much a seller can gain by deviating from the competitive equilibrium when buyers commit to the proportional rule.
\begin{theorem} \label{thm:constant_exploitability}
	Let $R^{\text{OPT}}_k$ be the optimal revenue of seller $k$ if all the other sellers choose their competitive equilibrium allocation, and $R^*_k$ be its revenue at any competitive equilibrium. Then
	$
	R_k^{\text{OPT}} \leq 5 R_k^*.
	$
\end{theorem}

\begin{proof}
	For any $x + c > 0$, it holds that
	$
		\frac{x}{x + c} \leq \ln(x + c) - \ln c.
	$
	We also have the following inequality:
	$
		\frac{x}{x + c}  \leq  1 = \ln(x + c) - \ln(x + c) + 1.
	$
	For each buyer $i$, let $w_i^*$ be the total valuation received by it from all sellers other than $k$ at some competitive equilibrium, and $u_i^*$ be the valuation received from seller $k$.
	Define $u_i = \sum_{j \in J_k} v_{i, j} x_{i, j}$.
	If all sellers other than $k$ choose their competitive equilibrium allocation, seller $k$'s revenue can be bounded as follows:
	\begin{align*}
	 	R_k(u, w^*)
		= 
		\sum_i B_i \cdot \frac{u_i}{u_i + w_i^*}
		\leq &
		\sum_i B_i  \ln(u_i + w_i^*)
		+ \sum_i B_i \cdot  \min \bigparen{-\ln w_i^*, -\ln(u_i + w_i^*) + 1}
		\\
		\leq &
		\sum_i B_i  \ln(u_i^* + w_i^*)
		+ \sum_i B_i \cdot  \min \bigparen{-\ln w_i^*, -\ln(u_i + w_i^*) + 1}
		\\
		= &
		\sum_i B_i \cdot \min \bigparen{ \ln \bigparen{\frac{u_i^* + w_i^*}{w_i^*}}, \ln \bigparen{\frac{u_i^* + w_i^*}{u_i + w_i^*}} + 1 }
		\\
		\leq &
		\sum_i B_i \cdot \min \bigparen{ \frac{u_i^*}{w_i^*}, \frac{u_i^* + w_i^*}{u_i + w_i^*} },
	\end{align*}
	where the second inequality holds since $(u^*, w^*)$ maximizes the EG-objective, and the last one holds since $\ln x \leq x - 1$.
	
	Suppose that, for every buyer $i$, $u_i \geq \frac{u_i^*}{2}$, then $\frac{u_i^* + w_i^*}{u_i + w_i^*} \leq 2$. In this case,
	\begin{align*}
		R_k(u, w^*)
		\leq 
		\sum_i B_i \cdot \min \bigparen{\frac{u_i^*}{w_i^*}, 2}
		=
		\sum_i B_i \cdot \frac{u_i^*}{u_i^* + w_i^*} \cdot \min \bigparen{\frac{u_i^* + w_i^*}{w_i^*}, \frac{2 (u_i^* + w_i^*)}{u_i^*}}
		\leq 
		3 R_k(u^*, w^*).
	\end{align*}
	
	Recall that $R_k(u, w)$ is a concave function with respect to $u$.
	Let $u^{\text{OPT}}$ be the best response of seller $k$ against $w^*$. Then
	\begin{align*}
		R_k(u^{\text{OPT}}, w^*)
		\leq 
		2 R_k\bigparen{\frac{u^{\text{OPT}} + u^*}{2}, w^*} - R_k(u^*, w^*)
		\leq 5 R_k(u^*, w^*).
	\end{align*}
	
\end{proof}

Nonetheless, a constant incentive ratio might be insufficient to prevent sellers from deviation.
Consider a seller-driven dynamics as follows: at time $t$, seller $k$ allocates its items to buyers with $\{x_{i, j}(t)\}_{j \in J_k}$ and buyer $i$ receives utilities $\{u_i(k, t) = \sum_j v_{i, j} x_{i, j}(t)\}_{k \in K}$.
Buyers still update their budget allocations according to the proportional rule. This determines the revenue of each seller at time $t + 1$.
From the perspective of online optimization, at time $t$, seller $k$ receives a convex loss:
$
\sum_i B_i \cdot \frac{ \sum_{k' \neq k} u_i(k', t)}{u_i(k, t) + \sum_{k' \neq k} u_i(k', t)}.
$
The information required for each seller to optimize it (even in an online fashion) is almost minimal (only $B_i, i \in I$ is needed).\footnote{$u_i(k, t)$ is also required, but it is always known to the seller since in auto-bidding markets the valuations are typically produced by machine learning models of the seller. $B_i$ may be slightly harder to acquire. It could be reported in the financial report of the buyer, or learned from direct communication with its advertising team.}
Hence sellers are more prone to manipulate the pacing equilibrium.

\section{Seller Competition and Fairness Guarantee}
\label{sec:seller_competition_game}

In this section, we explore the consequence of deviations from the proportional dynamics.
Based on the analysis in previous sections, buyers have a strong tendency to follow the proportional rule.
Therefore we put the focus on the strategic competition among sellers.
In this section, we assume that there are at least two sellers owning items with $v_{i, j} > 0$ for each buyer $i$.\footnote{If there is only one seller who provides positive utility for a buyer, the seller can fully absorb the buyer's budget with an infinitesimal amount of items.}

\begin{definition} \label{def:seller_game_pr}
	A \textbf{seller competition game} is defined over a linear Fisher market among sellers.
	The strategy space of seller $k$ is the set of feasible allocations $x_{i, j} \in [0, 1], i \in I, j \in J_k$.
	The revenue of seller $k$ with allocation profile $x$ is given by
	$
	R_k(x) = \sum_i B_i \cdot  \frac{\sum_{j \in J_k} v_{i, j} x_{i, j}}{\sum_{j \in J_k} v_{i, j} x_{i, j} + \sum_{j \in J \setminus J_k} v_{i, j} x_{i, j}}.
	$
	Without loss of generality,\footnote{See Appendix \ref{app:wlog_seller_competition_game} for why the assumption is without loss of generality.} we only consider allocation profiles $x$ such that $ \forall i, k$, $\sum_{j \in J \setminus J_k} v_{i, j} x_{i, j} \geq \epsilon$ with a sufficiently small $\epsilon > 0$.
	With this assumption, the revenue is always continuous and the set of allocation profiles is still compact and convex.
	A \textbf{pure Nash equilibrium} (PNE) of a seller competition game is an allocation profile $x$ such that each seller is best responding, i.e., $R_k(x) \geq R_k(x')$ for all $x'$ that agrees with $x$ for all $j \notin J_k$.
	
	In addition, note that if two allocation profiles produce the same utility profile $\{u_i(k)\}_{i, k}$, their corresponding game outcomes $\{R_{k}(x)\}_k$ are also identical.
	As a result, the strategy space of seller $k$ can be equivalently defined as the set of feasible utilities $\{u_i(k)\}_i$.
	The PNE could be defined in the same spirit, and we still restrict our attention on utility profiles such that $ \forall i, k$, $\sum_{k' \neq k} u_i(k') \geq \epsilon$ with a sufficiently small $\epsilon > 0$.
\end{definition}

\begin{theorem} \label{thm:unique_seller_competition_game}
	The seller competition game always admits a unique pure Nash equilibrium.
\end{theorem}

The proof is in Appendix \ref{app:unique_seller_competition_game}.
In general, the PNE of the seller competition game is not a competitive equilibrium.
But we have the following approximation guarantee for the NSW.

\begin{theorem}
	The NSW at the PNE of the seller competition game is at least $(1 - \Delta)$ of the optimal, where $\Delta := \max_{i, k} \frac{B_{i}(k, *)}{B_i}$ is the maximum fraction of budget earned by a single seller for a single buyer.
\end{theorem}
\begin{proof}
	Let $u_i(k, \dagger)$ be the utility received by buyer $i$ from seller $k$ at the PNE, and $u_i(k, *)$ be the utility at the competitive equilibrium.
	For the convenience of notation, in this proof we let $u_i(\dagger) = \sum_k u_i(k, \dagger)$ and $u_i(*) = \sum_k u_i(k, *)$.
	
	By first-order optimality of each seller's revenue at PNE, we have
	\begin{equation} \label{eqn:pne}
		\sum_i B_i \cdot \frac{u_i(\dagger) - u_i(k, \dagger)}{(u_i(\dagger)^2)}
		\cdot \bigparen{ u_i(k, *) - u_i(k, \dagger) } \leq 0, \forall i.
	\end{equation}
	
	By the definition of $\Delta$, we have
	\begin{align} \label{eqn:delta}
		\begin{split}
			\sum_k B_i(k, \dagger) B_i(k, *)
			\leq
			\bigparen{\sum_k B_i(k, \dagger)} \cdot  \max_{k} \bigbraces{ B_i(k, *) }
			 \leq
			\Delta B_i^2.
		\end{split}
	\end{align}
	
	We first consider the arithmetic mean of the buyer-wise ratio $\frac{u_i(*)}{u_i(\dagger)}$, which can be decomposed into two parts:
	\begin{align*}
		\sum_i B_i \cdot \frac{u_i(*)}{u_i(\dagger)}
		=
		\sum_{i, k} B_i \cdot \frac{u_i(k, *)}{u_i(\dagger)}
		=
		\sum_{k, i} B_i \cdot \frac{u_i(k, *)}{u_i(\dagger)} \cdot \bigparen{1 - \frac{u_i(k, \dagger)}{u_i(\dagger)} }
		+
		\sum_{i, k} B_i \cdot \frac{u_i(k, *)}{u_i(\dagger)} \cdot  \frac{u_i(k, \dagger)}{u_i(\dagger)}.
	\end{align*}
	Using inequality (\ref{eqn:pne}), the first part can be bounded as follows:
	\begin{align*}
		\sum_k \sum_i B_i \cdot \frac{u_i(k, *)}{u_i(\dagger)} \cdot \bigparen{1 - \frac{u_i(k, \dagger)}{u_i(\dagger)} }
		= &
		\sum_k \sum_i B_i \cdot  \frac{u_i(\dagger) - u_i(k, \dagger)}{(u_i(\dagger)^2)} \cdot u_i(k, *)
		\\
		\leq &
		\sum_k \sum_i B_i \cdot  \frac{u_i(\dagger) - u_i(k, \dagger)}{(u_i(\dagger)^2)} \cdot u_i(k, \dagger)
		\\
		\leq &
		\sum_i B_i \sum_k \frac{u_i(k, \dagger)}{u_i(\dagger)}
		= B.
	\end{align*}
	For the second part, note that at either equilibrium, buyer's bang-per-bucks are equalized among sellers: $
	\frac{u_i(k, \dagger)}{u_i(\dagger)} = \frac{B_i(k, \dagger)}{B_i},
	\frac{u_i(k, *)}{u_i(*)} = \frac{B_i(k, *)}{B_i}$.
	Combined with inequality (\ref{eqn:delta}), we have
	\begin{align*}
		\sum_i \sum_k B_i \cdot \frac{u_i(k, *)}{u_i(\dagger)} \cdot  \frac{u_i(k, \dagger)}{u_i(\dagger)} 
		= 
		\sum_k \sum_i  \frac{u_i(*)}{B_i u_i(\dagger)} \cdot \sum_k B_i(k, \dagger) B_i(k, *)
		\leq
		\Delta \sum_i B_i \cdot \frac{u_i(*)}{u_i(\dagger)}.
	\end{align*}
	
	Putting both parts together and rearranging the terms gives
	\begin{displaymath}
		(1 - \Delta) \sum_i B_i \cdot \frac{u_i(*)}{u_i(\dagger)} \leq B.
	\end{displaymath}
	
	By the weighted AM-GM inequality, the ratio of NSW at each equilibrium is bounded by:
	\begin{align*}
		\prod_i \bigparen{\frac{u_i(*)}{u_i(\dagger)}}^{\frac{B_i}{B}}
		& \leq
		\frac{1}{B} \cdot \sum_i B_i\cdot \frac{u_i(*)}{u_i(\dagger)}
		\leq
		\frac{1}{1 - \Delta}.
	\end{align*}
\end{proof}

The approximation ratio improves as $\Delta$ becomes smaller.
Intuitively, $\Delta$ describes how \textit{monopolized} the market is.
For instance, if $\Delta = 0.9$, then at the competitive equilibrium some buyer $i$ will spend 90\% of its budget on a single seller $j$, and other sellers may not have enough power to compete with seller $j$ for buyer $i$.
As a result, seller $j$ may raise the price for buyer $i$ and spare its items for other buyers' budgets, which hurts the utility of buyer $i$ and the overall fairness.
Conversely, with a small $\Delta$, no seller earns a significant portion of buyers' budgets and the market is more balanced and competitive.
In this case, even if sellers deviate, the end result will remain close to the competitive equilibrium.

\section{Conclusion}
This paper extends proportional dynamics to auto-bidding auction markets, a first generalization to the one-seller multi-item scenario.
While maintaining a consistent convergence to the competitive equilibrium, we find that buyers adhere more to the proportional rule, while sellers can easily deviate for more revenue. This leads to a novel seller-side game with a unique pure Nash equilibrium, which still ensures good fairness in competitive markets. We hope that our work could help people understand the strategic behaviors in online advertising markets, and provide valuable insights into modern price-forming dynamics.


%


\bibliographystyle{abbrvnat}
\bibliography{bibliography}

\begin{thebibliography}{30}
\providecommand{\natexlab}[1]{#1}
\providecommand{\url}[1]{\texttt{#1}}
\expandafter\ifx\csname urlstyle\endcsname\relax
  \providecommand{\doi}[1]{doi: #1}\else
  \providecommand{\doi}{doi: \begingroup \urlstyle{rm}\Url}\fi

\bibitem[Aggarwal et~al.(2019)Aggarwal, Badanidiyuru, and
  Mehta]{aggarwal2019autobidding}
G.~Aggarwal, A.~Badanidiyuru, and A.~Mehta.
\newblock Autobidding with constraints.
\newblock In \emph{International Conference on Web and Internet Economics},
  pages 17--30. Springer, 2019.

\bibitem[Balseiro et~al.(2021)Balseiro, Deng, Mao, Mirrokni, and
  Zuo]{balseiro2021robust}
S.~Balseiro, Y.~Deng, J.~Mao, V.~Mirrokni, and S.~Zuo.
\newblock Robust auction design in the auto-bidding world.
\newblock \emph{Advances in Neural Information Processing Systems},
  34:\penalty0 17777--17788, 2021.

\bibitem[Birnbaum et~al.(2011)Birnbaum, Devanur, and
  Xiao]{birnbaum2011distributed}
B.~Birnbaum, N.~R. Devanur, and L.~Xiao.
\newblock Distributed algorithms via gradient descent for {F}isher markets.
\newblock In \emph{Proceedings of the 12th ACM conference on Electronic
  commerce}, pages 127--136, 2011.

\bibitem[Br{\^a}nzei et~al.(2021)Br{\^a}nzei, Devanur, and
  Rabani]{branzei2021proportional}
S.~Br{\^a}nzei, N.~Devanur, and Y.~Rabani.
\newblock Proportional dynamics in exchange economies.
\newblock In \emph{Proceedings of the 22nd ACM Conference on Economics and
  Computation}, pages 180--201, 2021.

\bibitem[Br{\^a}nzei et~al.(2022)Br{\^a}nzei, Gkatzelis, and
  Mehta]{branzei2022nash}
S.~Br{\^a}nzei, V.~Gkatzelis, and R.~Mehta.
\newblock Nash social welfare approximation for strategic agents.
\newblock \emph{Operations Research}, 70\penalty0 (1):\penalty0 402--415, 2022.

\bibitem[Chen et~al.(2011)Chen, Deng, and Zhang]{chen2011profitable}
N.~Chen, X.~Deng, and J.~Zhang.
\newblock How profitable are strategic behaviors in a market?
\newblock In \emph{Algorithms--ESA 2011: 19th Annual European Symposium,
  Saarbr{\"u}cken, Germany, September 5-9, 2011. Proceedings 19}, pages
  106--118. Springer, 2011.

\bibitem[Chen and Teng(2009)]{chen2009spending}
X.~Chen and S.-H. Teng.
\newblock Spending is not easier than trading: on the computational equivalence
  of {F}isher and {A}rrow-{D}ebreu equilibria.
\newblock In \emph{International Symposium on Algorithms and Computation},
  pages 647--656. Springer, 2009.

\bibitem[Cheung et~al.(2018{\natexlab{a}})Cheung, Cole, and
  Tao]{cheung2018dynamics}
Y.~K. Cheung, R.~Cole, and Y.~Tao.
\newblock Dynamics of distributed updating in {F}isher markets.
\newblock In \emph{Proceedings of the 2018 ACM Conference on Economics and
  Computation}, pages 351--368, 2018{\natexlab{a}}.

\bibitem[Cheung et~al.(2018{\natexlab{b}})Cheung, Hoefer, and
  Nakhe]{cheung2018tracing}
Y.~K. Cheung, M.~Hoefer, and P.~Nakhe.
\newblock Tracing equilibrium in dynamic markets via distributed adaptation.
\newblock \emph{arXiv preprint arXiv:1804.08017}, 2018{\natexlab{b}}.

\bibitem[Cole and Fleischer(2008)]{cole2008fast}
R.~Cole and L.~Fleischer.
\newblock Fast-converging tatonnement algorithms for one-time and ongoing
  market problems.
\newblock In \emph{Proceedings of the fortieth annual ACM symposium on Theory
  of computing}, pages 315--324, 2008.

\bibitem[Cole et~al.(2013)Cole, Gkatzelis, and Goel]{cole2013mechanism}
R.~Cole, V.~Gkatzelis, and G.~Goel.
\newblock Mechanism design for fair division: allocating divisible items
  without payments.
\newblock In \emph{Proceedings of the fourteenth ACM conference on Electronic
  commerce}, pages 251--268, 2013.

\bibitem[Conitzer et~al.(2022{\natexlab{a}})Conitzer, Kroer, Panigrahi,
  Schrijvers, Stier-Moses, Sodomka, and Wilkens]{conitzer2022pacing}
V.~Conitzer, C.~Kroer, D.~Panigrahi, O.~Schrijvers, N.~E. Stier-Moses,
  E.~Sodomka, and C.~A. Wilkens.
\newblock Pacing equilibrium in first price auction markets.
\newblock \emph{Management Science}, 68\penalty0 (12):\penalty0 8515--8535,
  2022{\natexlab{a}}.

\bibitem[Conitzer et~al.(2022{\natexlab{b}})Conitzer, Kroer, Sodomka, and
  Stier-Moses]{conitzer2022multiplicative}
V.~Conitzer, C.~Kroer, E.~Sodomka, and N.~E. Stier-Moses.
\newblock Multiplicative pacing equilibria in auction markets.
\newblock \emph{Operations Research}, 70\penalty0 (2):\penalty0 963--989,
  2022{\natexlab{b}}.

\bibitem[Deng et~al.(2021)Deng, Mao, Mirrokni, and Zuo]{deng2021towards}
Y.~Deng, J.~Mao, V.~Mirrokni, and S.~Zuo.
\newblock Towards efficient auctions in an auto-bidding world.
\newblock In \emph{Proceedings of the Web Conference 2021}, pages 3965--3973,
  2021.

\bibitem[Despotakis et~al.(2021)Despotakis, Ravi, and
  Sayedi]{despotakis2021first}
S.~Despotakis, R.~Ravi, and A.~Sayedi.
\newblock First-price auctions in online display advertising.
\newblock \emph{Journal of Marketing Research}, 58\penalty0 (5):\penalty0
  888--907, 2021.

\bibitem[Dvijotham et~al.(2022)Dvijotham, Rabani, and
  Schulman]{dvijotham2022convergence}
K.~Dvijotham, Y.~Rabani, and L.~J. Schulman.
\newblock Convergence of incentive-driven dynamics in {F}isher markets.
\newblock \emph{Games and Economic Behavior}, 134:\penalty0 361--375, 2022.

\bibitem[Eisenberg and Gale(1959)]{eisenberg1959consensus}
E.~Eisenberg and D.~Gale.
\newblock Consensus of subjective probabilities: The pari-mutuel method.
\newblock \emph{The Annals of Mathematical Statistics}, 30\penalty0
  (1):\penalty0 165--168, 1959.

\bibitem[Feldman et~al.(2005)Feldman, Lai, and Zhang]{feldman2005price}
M.~Feldman, K.~Lai, and L.~Zhang.
\newblock A price-anticipating resource allocation mechanism for distributed
  shared clusters.
\newblock In \emph{Proceedings of the 6th ACM conference on Electronic
  commerce}, pages 127--136, 2005.

\bibitem[Gao et~al.(2021)Gao, Peysakhovich, and Kroer]{gao2021online}
Y.~Gao, A.~Peysakhovich, and C.~Kroer.
\newblock Online market equilibrium with application to fair division.
\newblock \emph{Advances in Neural Information Processing Systems},
  34:\penalty0 27305--27318, 2021.

\bibitem[Goktas et~al.(2023)Goktas, Zhao, and Greenwald]{goktas2023t}
D.~Goktas, J.~Zhao, and A.~Greenwald.
\newblock T{\^a}tonnement in homothetic {F}isher markets.
\newblock In \emph{Proceedings of the 24th ACM Conference on Economics and
  Computation}, pages 760--781, 2023.

\bibitem[Golrezaei et~al.(2021)Golrezaei, Lobel, and
  Paes~Leme]{golrezaei2021auction}
N.~Golrezaei, I.~Lobel, and R.~Paes~Leme.
\newblock Auction design for {ROI}-constrained buyers.
\newblock In \emph{Proceedings of the Web Conference 2021}, pages 3941--3952,
  2021.

\bibitem[Kelly(1997)]{kelly1997charging}
F.~Kelly.
\newblock Charging and rate control for elastic traffic.
\newblock \emph{European transactions on Telecommunications}, 8\penalty0
  (1):\penalty0 33--37, 1997.

\bibitem[Li and Tang(2023)]{li2023vulnerabilities}
J.~Li and P.~Tang.
\newblock Vulnerabilities of single-round incentive compatibility in
  auto-bidding: Theory and evidence from {ROI}-constrained online advertising
  markets, 2023.

\bibitem[Orlin(2010)]{orlin2010improved}
J.~B. Orlin.
\newblock Improved algorithms for computing {F}isher's market clearing prices:
  Computing {F}isher's market clearing prices.
\newblock In \emph{Proceedings of the forty-second ACM symposium on Theory of
  computing}, pages 291--300, 2010.

\bibitem[Paes~Leme et~al.(2020)Paes~Leme, Sivan, and Teng]{paes2020competitive}
R.~Paes~Leme, B.~Sivan, and Y.~Teng.
\newblock Why do competitive markets converge to first-price auctions?
\newblock In \emph{Proceedings of The Web Conference 2020}, pages 596--605,
  2020.

\bibitem[Shapley and Shubik(1977)]{shapley1977trade}
L.~Shapley and M.~Shubik.
\newblock Trade using one commodity as a means of payment.
\newblock \emph{Journal of political economy}, 85\penalty0 (5):\penalty0
  937--968, 1977.

\bibitem[Shmyrev(2009)]{shmyrev2009algorithm}
V.~I. Shmyrev.
\newblock An algorithm for finding equilibrium in the linear exchange model
  with fixed budgets.
\newblock \emph{Journal of Applied and Industrial Mathematics}, 3:\penalty0
  505--518, 2009.

\bibitem[Walras(1900)]{walras1900elements}
L.~Walras.
\newblock \emph{{\'E}l{\'e}ments d'{\'e}conomie politique pure: ou, Th{\'e}orie
  de la richesse sociale}.
\newblock F. Rouge, 1900.

\bibitem[Wu and Zhang(2007)]{wu2007proportional}
F.~Wu and L.~Zhang.
\newblock Proportional response dynamics leads to market equilibrium.
\newblock In \emph{Proceedings of the thirty-ninth annual ACM symposium on
  Theory of computing}, pages 354--363, 2007.

\bibitem[Zhang(2011)]{zhang2011proportional}
L.~Zhang.
\newblock Proportional response dynamics in the {F}isher market.
\newblock \emph{Theoretical Computer Science}, 412\penalty0 (24):\penalty0
  2691--2698, 2011.

\end{thebibliography}

\clearpage
\newpage

\appendix

\section{An Example of Non-smoothness and Non-convexity of the Buyer-side Optimization Problem}
\label{app:example_buyer_optimization}

\begin{example} \label{example:non_convexity_buyer_optimization}
	Consider a sub-market with two items and two buyers where $v_{1, 1} = v_{2, 2} = 2$ and $v_{1, 2} = v_{2, 1} = 1$.
	Suppose buyer 1's budget is 2.
	Let $f(b)$ be the value received by buyer 2 at the pacing equilibrium with budget $b$. Then
	\begin{displaymath}
		f(b) = \left\{
		\begin{array}{ll}
			\frac{6b}{b + 2}, & \text{if } b \in [0, 1]; \\
			2, & \text{if } b \in [1, 4]; \\
			3 - \frac{6}{b + 2}, & \text{if } b \in [4, +\infty).
		\end{array}
		\right.
	\end{displaymath}
	Suppose at the other sub-market the valuation is identical and buyer 1 also puts a budget of 2.
	Then the utility of buyer 2 in the budget allocation game will be $f(b) + f(B_2 - b)$, where $b$ is the amount of money it pays for one of the sub-markets.
\end{example}

\section{Incentive of Buyers in the One-seller One-item Setting}
\label{app:incentive_buyers_original}

When each seller owns only a single item, fixing other buyers' bid for each seller/item, the optimization problem of buyer $i$ is given by the following convex program:
\begin{align*}
	\max \quad &
	\sum_k \frac{B_i(k)}{B_i(k) + \sum_{i' \neq i} B_{i'}(k)} \cdot  v_{i, k} \\
	\text{s.t.} \quad &
	\sum_k B_i(k) \leq B_i \\
	& B_i(k) \geq 0, \forall k.
\end{align*}
The optimal solution is generally \textit{not} the proportional rule, and the only information required to solve the program is its own valuation of items and the fraction of each item bought at the last time step (from which $ \sum_{i' \neq i} B_{i'}(k)$ can be deduced easily).
Since online convex optimization is well-studied in the literature and widely implemented in real-world applications, buyers can easily deviate from the proportional rule in this case.

\section{Pacing Equilibrium with Additive Boosts}
\label{app:pacing_equilibrium_with_boosts}

\begin{lemma} \label{thm:pacing_equilibrium_with_boosts}
	The pacing equilibrium with additive boosts is a fixed point of proportional dynamics with additive boosts, which is given by the solution of the following modified EG program:
	\begin{align*}
		\max \quad & \sum_i B_i \ln \bigparen{ \sum_j v_{i, j} x_{i, j} } + \sum_{i, j} c_{i, j} x_{i, j}
		\\
		\text{s.t.} \quad &
		\sum_i x_{i, j} \leq 1, \forall j;
		\\
		& x_{i, j} \geq 0, \forall i, j.
	\end{align*}
\end{lemma}
\begin{proof}
	Let  $p_j$ be the dual multiplier associated with the capacity constraint of item $j$, and $u_i = \sum_j v_{i, j} x_{i, j}$ be the total value acquired by buyer $i$.
	By KKT conditions, the optimal solution satisfies
	\begin{displaymath}
		- \frac{B_i v_{i, j}}{u_i} - c_{i, j} + p_j \geq 0.
	\end{displaymath}
	It can be verified that each buyer spends $\sum_j (p_j - c_{i, j}) x_{i, j} = B_i$ in total and its bang-per-buck $\frac{p_j - c_{i, j}}{v_{i, j}}$ is equalized across items.
\end{proof}

\begin{corollary}
	By strict convexity of the objective w.r.t. $u_i$, buyer's utility at equilibrium is unique.
\end{corollary}

\section{References for Lemma \ref{lemma:monotonicity_budget}}
\label{app:budget_monotonicity}

The first half of the lemma is a restatement of a result in \cite{chen2011profitable}, which states that truthfully reporting the budget is a dominant strategy for fixed linear utility functions.

The second half of the lemma is a corollary of Lemma 1 in \cite{conitzer2022pacing}, which states that the pacing multiplier (inverse of bang-per-bucks) profile at the pacing equilibrium buyer-wise dominates any budget-feasible multiplier profile.
If some buyer raises its budget, the old pacing equilibrium is still budget-feasible with the new budget profile, and it is thus buyer-wise dominated by the new pacing equilibrium w.r.t. the inverse of bang-per-bucks.

\section{Proof of Theorem \ref{thm:convergence_with_boosts}}
\label{app:convergence_with_boosts}

\begin{proof}
	Following the definition of $\Phi$ in the proof of Theorem \ref{thm:convergence_non_personalized_pd},
	we still have
	\begin{align*}
		\Phi(t) - \Phi(t + 1)
		=
		\sum_{i, k} B_i(k, *) \ln \frac{u_i(*)}{u_i(t)}
		+
		\sum_{i, k}  B_i(k, *) \ln \bigparen{\frac{B_i(k, *)}{B_i(k, t) \cdot \frac{u_i(k, *)}{u_i(k, t)}}}.
	\end{align*}
	
	The first term is now lower bounded by the modified EG-objective (Lemma \ref{thm:pacing_equilibrium_with_boosts}):
	\begin{displaymath}
		\sum_{i, k} B_i(k, *) \ln \frac{u_i(*)}{u_i(t)}
		\geq
		\sum_{i, j} c_{i, j} \bigparen{x_{i, j}(t) - x_{i, j}(*)}
		:= C.
	\end{displaymath}
	Note that $C$ depends on $x(t)$ and is not a constant.
	
	For the second term, again let $\alpha_i(k, t) = \frac{B_i(k, t)}{u_i(k, t)}$ and we have
	\begin{align*}
		\tilde{B}
		= &
		\sum_{k, i} \alpha_i(k, t) \sum_{j \in J_k} v_{i, j} x_{i, j}(*) \\
		\leq &
		\sum_{k, i} \sum_{j \in J_k} \bigparen{ \alpha_i(k, t) v_{i, j} + c_{i, j} } x_{i, j}(t)
		- \sum_{k, i} \sum_{j \in J_k} c_{i, j} x_{i, j}(*) \\
		= & B + C,
	\end{align*}
	where $B$ stands for the total budget of all buyers, and the inequality holds since, at the pacing equilibrium with additive boosts, every item is sold to the buyer with the maximum boosted bid $\alpha_i v_{i, j} + c_{i, j}$.
	
	Consider the function
	\begin{displaymath}
		f(C) := C + \sum_{i, k} B_i(k, *) \ln \bigparen{ \frac{B_i(k, *)}{\frac{B + C}{\tilde{B}} \cdot B_i(k, t) \cdot \frac{u_i(k, *)}{u_i(k, t)}} }.
	\end{displaymath}
	Take the gradient:
	\begin{displaymath}
		\fdiff{f}{C} = \frac{C}{B + C}.
	\end{displaymath}
	Since $B + C > 0$, $f$ is minimized at $C = 0$.
	The rest part is identical to the proof of Theorem \ref{thm:convergence_non_personalized_pd}.
\end{proof}

\section{Definition of Seller Competition Game}
\label{app:wlog_seller_competition_game}

The assumption $ \forall i, k$, $\sum_{j \in J \setminus J_k} v_{i, j} x_{i, j} \geq \epsilon$ basically requires that there are at least two different sellers competing for each buyer. It is without loss of generality since as  $\sum_{j \in J \setminus J_k} v_{i, j} x_{i, j}$ goes to zero, seller $k$ tends to allocate a smaller fraction of its items to buyer $i$, which is enough to absorb almost all the buyer's budget as the competition is nearly non-existent.
This means $u_i(k)$ also goes to zero.
However, the marginal gain to join the competition for buyer $i$ becomes larger and larger for other sellers and eventually someone will do so.
Therefore such states cannot form a Nash equilibrium.

\section{Proof of Theorem \ref{thm:unique_seller_competition_game}}
\label{app:unique_seller_competition_game}

\begin{proof}
	Let $U$ be the set of all implementable buyer utility profiles $\{u_i(k)\}_{i, k}$ such that $\sum_{k \neq k} u_i(k') \geq \epsilon$ for every $i$ and $k$.
	$U$ is compact and convex.
	Define function $f: U \times U \rightarrow \realnum$ as:
	\begin{displaymath}
		f(u, w) = \sum_{i, k} B_i \cdot \frac{u_i(k)}{u_i(k) + \sum_{k' \neq k} w_i(k')}, \forall u, w \in U.
	\end{displaymath}
	$f$ is strictly concave w.r.t. $u$ (since $\sum_{k \neq k} u_i(k') \geq \epsilon$) and convex w.r.t. $w$.
	Let $B = \sum_i B_i$, and we have
	\begin{align*}
		\max_{u \in U}  f(u, w) \geq f(w, w) = B
		\quad
		\Rightarrow 
		\quad
		\min_{w \in U} \max_{u \in U}  f(u, w) \geq B;
		\\
		\min_{w \in U}  f(u, w) \leq f(u, u) = B
		\quad
		\Rightarrow
		\quad
		\max_{u \in U} \min_{w \in U}  f(u, w) \leq B.
	\end{align*}
	By minimax theorem, $
	\max_{u \in U} \min_{w \in U} f(u, w)
	=
	\min_{w \in U} \max_{u \in U}  f(u, w) = B,
	$
	and there exists $(u^*, w^*)$ that satisfies $f(u^*, w^*) = B$ and
	\begin{displaymath}
		u^* \in \argmax_{u \in U} f(u, w^*), \quad
		w^* \in \argmin_{w \in U} f(u^*, w).
	\end{displaymath}
	By strict concavity of $f(u, w)$ w.r.t. $u$, $u^* = w^*$. 
	Then $u^*$ is a Nash equilibrium since $u^* \in \argmax_{u \in U} f(u, u^*)$.
	
	Suppose $u' \neq u^*$ is also a pure Nash equilibrium.
	By strict concavity of $f$ w.r.t. $u$, $f(u^*, u') < f(u', u') = B$, this contradicts the fact that $u^* \in \argmin_{u \in U} f(u^*, u)$.
	Therefore $u^*$ is the unique pure Nash equilibrium.
\end{proof}





\end{document}